\documentclass[a4paper,UKenglish,cleveref, autoref, thm-restate]{lipics-v2021}

\hideLIPIcs  

\title{On Finding \texorpdfstring{$\ell$}{l}-th Smallest Perfect Matchings} 

\author{Nicolas {El Maalouly}}{Department of Computer Science, ETH Zurich, Switzerland }{nicolas_m_7@hotmail.com}{0000-0002-1037-0203}{} 
\author{Sebastian {Haslebacher}}{Department of Computer Science, ETH Zurich, Switzerland }{sebastian.haslebacher@inf.ethz.ch}{0000-0003-3988-3325}{}
\author{Adrian {Taubner}}{Department of Computer Science, ETH Zurich, Switzerland }{ataubner@inf.ethz.ch}{}{} 
\author{Lasse {Wulf}}{Section of Algorithms, Logic, and Graphs, Technical University of Denmark, Denmark \and The IT University, Copenhagen, Denmark}{lawu@dtu.dk}{0000-0001-7139-4092}{Supported by Carlsberg Foundation CF21-0302 ``Graph Algorithms with Geometric Applications''} 

\authorrunning{N. El Maalouly, S. Haslebacher, A. Taubner, L. Wulf}

\Copyright{Nicolas El Maalouly, Sebastian Haslebacher, Adrian Taubner, Lasse Wulf} 



\category{} 

\relatedversion{} 

\nolinenumbers 

\EventEditors{Anne Benoit, Haim Kaplan, Sebastian Wild, and Grzegorz Herman}
\EventNoEds{4}
\EventLongTitle{33rd Annual European Symposium on Algorithms (ESA 2025)}
\EventShortTitle{ESA 2025}
\EventAcronym{ESA}
\EventYear{2025}
\EventDate{September 15--17, 2025}
\EventLocation{Warsaw, Poland}
\EventLogo{}
\SeriesVolume{351}
\ArticleNo{17}

\usepackage{mathtools}
\usepackage{amsmath, amssymb, amsfonts}

\newcommand{\NP}{\textsf{NP}}
\renewcommand{\P}{\textsf{P}}
\newcommand{\BPP}{\textsf{BPP}}
\newcommand{\RP}{\textsf{RP}}

\newcommand{\N}{\mathbb{N}}

\newcommand{\Z}{\mathbb{Z}}
\newcommand{\bigO}{\mathcal{O}}

\DeclarePairedDelimiter\abs{\lvert}{\rvert}

\newcommand{\nequiv}{\not\equiv}

\renewcommand{\epsilon}{\ensuremath\varepsilon}

\renewcommand{\phi}{\ensuremath{\varphi}}
\newcommand{\MWPM}{\textrm{MWPM}}
\newcommand{\EWPM}{\textrm{EWPM}}
\newcommand{\EM}{\textrm{EM}}
\newcommand{\lBPM}{$\textrm{SPM}_{\ell}$}

\newcommand{\BCPM}{\textrm{BCPM}}

\newcommand{\SOC}{\textrm{SOC}}

\newcommand{\ECS}{\textrm{ECS}}

\newcommand{\PMs}{\mathcal{M}}
\newcommand{\A}{\mathcal{A}}

\begin{document}

\maketitle

\begin{abstract}
    Given an undirected weighted graph $G$ and an integer $k$, 
    Exact-Weight Perfect Matching (\EWPM) 
    is the problem of finding a perfect matching of weight exactly $k$ in $G$. 
    In this paper, we study \EWPM\ and its variants. 
    The \EWPM\ problem is famous, since in the case of unary encoded weights, 
    Mulmuley, Vazirani, and Vazirani showed almost 40 years ago that the problem can be solved in randomized polynomial time.
    However, up to this date no derandomization is known.

    Our first result is a simple deterministic algorithm for \EWPM\ that runs in time $n^{\bigO(\ell)}$, 
    where $\ell$ is the number of distinct weights that perfect matchings in $G$ can take. 
    In fact, we show how to find an $\ell$-th smallest perfect matching in any weighted graph (even if the weights are encoded in binary, in which case \EWPM\ in general is known to be \NP-complete) in time $n^{\bigO(\ell)}$ for any integer $\ell$. Similar next-to-optimal variants have also been studied recently for the shortest path problem.

    For our second result, we extend the list of problems that are known to be equivalent to  {\EWPM}. 
    We show that {\EWPM} is equivalent under a weight-preserving reduction to the Exact Cycle Sum problem (\ECS) in undirected graphs with a conservative (i.e.\ no negative cycles) weight function. 
    To the best of our knowledge, we are the first to study this problem. As a consequence, the latter problem is contained in \RP\ if the weights are encoded in unary.
    Finally, we identify a special case of {\EWPM}, called {\BCPM}, which was recently studied by El Maalouly, Steiner and Wulf. 
    We show that \BCPM\ is equivalent under a weight-preserving transformation to another problem recently studied by Schlotter and Sebő as well as Geelen and Kapadia: the  Shortest Odd Cycle problem (\SOC) in undirected graphs with conservative weights.
    Finally, our $n^{\bigO(\ell)}$ algorithm works in this setting as well, identifying a tractable special case of \SOC, \BCPM, and \ECS.
\end{abstract}

\section{Introduction}
\label{sec:introduction}

A big open question in complexity theory concerns the role of randomness in devising polynomial-time algorithms. While it is conjectured that every problem in \RP\ (and even \BPP) should also lie in \P\ (see e.g.\ Chapter~20 in~\cite{aroraComputationalComplexityModern2009} for more details), 
there exist several natural problems for which efficient randomized algorithms, but no efficient deterministic algorithms, are known. Thus, instead of attacking the conjecture \RP\ $=$ \P\ directly, it might be more reasonable to first try to devise efficient deterministic algorithms for those problems.

A famous example of a problem in \RP, that is not known to lie in \P, is the unary version of Exact-Weight Perfect Matching (\EWPM): Given an undirected weighted graph $G$ and an integer $k$, \EWPM\ is the problem of deciding whether $G$ admits a perfect matching of weight exactly $k$. Note that this problem is often called Exact Matching (\EM) if we allow only weights $0$ and $1$, and it is not hard to see (by replacing single edges by paths corresponding to their weight) that unary \EWPM\ is equivalent to \EM. 

\EWPM\ was introduced by Papadimitriou and Yannakakis~\cite{papadimitriouComplexityRestrictedSpanning1982} in 1982, and it is known to be \NP-hard in general (for a short proof, see~\cite{gurjarPlanarizingGadgetsPerfect2012}). Papadimitriou and Yannakakis also conjectured that \EWPM\ remains \NP-hard if weights are encoded in unary. However, Mulmuley, Vazirani and Vazirani~\cite{mulmuleyMatchingEasyMatrix1987} proved in 1987 that unary \EWPM\ is in \RP. This makes it very unlikely for the problem to be \NP-hard. Hence, the question became whether one could also find an efficient deterministic algorithm for unary \EWPM. 

This question has been open for almost 40 years now. In fact, despite considerable effort, no deterministic algorithm is known that runs in time that is subexponential in the number of vertices of the graph. Instead, progress has been made mostly by focusing on restricted graph classes, approximation variants of \EWPM, and finding other natural problems that are equivalent to unary \EWPM~(see Section~\ref{ssec:further_related_work} for more details). 

Our main result is a new parameterized algorithm for \EWPM\ that works on general graphs, even if weights are encoded in binary: Assume that $k_1 < k_2 <  \dots < k_\ell$ are all the distinct weights up to $k_\ell$ that perfect matchings in $G$ take. We prove that it is possible to find a perfect matching of weight $k_\ell$ in time $n^{\bigO(\ell)}$. Our algorithm is very simple and relies on our main theorem (Theorem~\ref{theorem:l-best}), which says that there must exist $2(\ell - 1)$ edges such that any minimum-weight perfect matching that includes all of these $2(\ell - 1)$ edges must have weight $k_\ell$.
The obtained bound of $2(\ell - 1)$ is tight.

\subsection{Next-to-Optimal Problems}
\label{ssec:next-to-optimal_problems}

Our main result can be equivalently stated in a different way: Consider the following problem, which we call
$\ell$-th Smallest Perfect Matching (\lBPM): Given a weighted graph $G$, an integer $k$, and a natural number $\ell$, decide whether $G$ contains an $\ell$-th smallest perfect matching of weight exactly $k$. Here, the notion of $\ell$-th smallest perfect matchings is defined recursively: $1$st-smallest perfect matchings are exactly the minimum-weight perfect matchings, and for $\ell \geq 2$, $\ell$-th smallest perfect matchings are the perfect matchings of minimum weight among all perfect matchings that are not $\ell'$-th smallest for any $\ell' < \ell$. In this language, our result is an algorithm that solves \lBPM\ in time $n^{\bigO(\ell)}$.

Similar next-to-optimal problems have been studied in the literature before.
In particular, there is a line of work on finding next-to-shortest paths in directed or undirected graphs. In these problems, we want to find a next-to-shortest $(s, t)$-path in a weighted graph $G$, i.e.\ the path should have minimum weight among all $(s, t)$-paths that are not shortest paths. This problem is \NP-hard on directed graphs~\cite{lalgudiComputingStrictlysecondShortest1997}, but admits polynomial-time algorithms on undirected graphs~\cite{kaoQuadraticAlgorithmFinding2011, krasikovFindingNextshortestPaths2004, liImprovedAlgorithmFinding2006, wuSimplerMoreEfficient2013, zhangNextShortestPathUndirected2012}. Note that the \NP-hardness result in~\cite{lalgudiComputingStrictlysecondShortest1997} heavily relies on allowing edges to have weight zero. Indeed, the so-called detour problem --- decide whether all $(s, t)$-shortest paths in a given unweighted directed graph have the same length or not --- is very similar but does not allow for zero-weight edges, and its complexity status is still unresolved~\cite{hatzelSimplerFasterAlgorithms2023}.

Instead of asking for a perfect matching of $\ell$-th smallest weight, one could instead ask for $\ell$ distinct perfect matchings while minimizing their combined weight. 
This is a conceptually different problem which is equivalent to repeatedly ($\ell$ times) finding the best solution distinct from previously obtained solutions. Such $k$-best enumeration problems have been studied for a number of combinatorial optimization problems: Murty~\cite{murty1968algorithm} first showed a polynomial-time algorithm for the $k$-best bipartite perfect matching problem. Later, the time complexity was improved by Chegireddy and Hamacher \cite{chegireddy1987algorithms}.
For more information, we refer the reader to the survey by Eppstein~\cite{eppsteinKbestEnumeration2014}.

\subsection{Related Cycle Problems}
\label{ssec:related_cycle_problems}

Our second contribution in this paper is to strengthen the connection between \EWPM\ and related cycle problems, as we will outline here. 

Consider the following problem, which we call Exact Cycle Sum (\ECS): Given an undirected graph $G$ with conservative (i.e.\ no negative cycles) edge-weights and an integer $k$, decide whether $G$ contains a set of disjoint cycles of weight exactly $k$. A directed version of \ECS\ with general weights has been studied e.g.\ by Papadimitriou and Yannakakis~\cite{papadimitriouComplexityRestrictedSpanning1982}. Concretely, they proved that the directed version of \ECS\ is equivalent to \EWPM\ on bipartite graphs. We extend this result by proving that the undirected problem \ECS\ as defined above is equivalent to \EWPM\ on general graphs. Our reductions preserve the encoding of the weights. Hence, we obtain as a result that \ECS\ is in \RP\ if weights are encoded in unary, and \NP-complete if weights are encoded in binary.

Another interesting problem on cycles in graphs is the so-called Shortest Odd Cycle problem (\SOC): Given an undirected graph $G$ with conservative edge-weights and an integer $k$, decide whether $G$ contains a cycle of odd weight at most $k$. \SOC\ was recently studied by Schlotter and Seb\H o~\cite{schlotterOddPathsCycles2025} (they use slightly different but equivalent definition, see Section~\ref{ssec:conservative_weights_and_cycle_problems} for more details), who proved that it is equivalent to finding minimum-weight odd $T$-joins. They also point out that \SOC\ is contained in \RP\ for unary weights, due to a result by Geelen and Kapadia~\cite{geelenComputingGirthCogirth2018}, and therefore conjecture that the problem should be contained in \P\ as well. Interestingly, the algorithm in~\cite{geelenComputingGirthCogirth2018} follows from a reduction to a problem on perfect matchings that can then be solved using the techniques of Mulmuley, Vazirani, and Vazirani~\cite{mulmuleyMatchingEasyMatrix1987}. In fact, this problem on perfect matchings has independently been studied by El Maalouly, Steiner, and Wulf~\cite{elmaaloulyExactMatchingCorrect2023} who called it Bounded Correct-Parity Perfect Matching (\BCPM): Given a weighted graph $G$ and an integer $k$, decide whether $G$ contains a perfect matching $M$ of weight $w(M)$ at most $k$ and with $w(M) \equiv_2 k$ (where $\equiv_2$ is used to denote mod $2$ equivalence). El Maalouly, Steiner, and Wulf use \BCPM\ on bipartite graphs as a subroutine for a parameterized algorithm for \EM\ on bipartite graphs. To this end, they prove that \BCPM\ on bipartite graphs is in \P, by reducing it to the directed version of \SOC\ (which can be solved via dynamic programming). In this paper, we extend this set of results by proving that \BCPM\ on general graphs is equivalent to \SOC\ as defined above. Unfortunately, we are not able to give deterministic polynomial-time algorithms for either problem, but any such deterministic polynomial-time algorithm for \SOC\ or \BCPM\ would immediately imply that the parameterized algorithm from~\cite{elmaaloulyExactMatchingCorrect2023} also works on general graphs. Furthermore, our results highlight that the two equivalent problems \SOC\ and \BCPM\ can be seen as a stepping stone for solving unary $\EWPM$, and that it might be reasonable to tackle them first. 

Concretely, some interesting further questions that remain unresolved are: Can we solve \SOC\ and \BCPM\ in deterministic polynomial time? The conjecture \RP\ $=$ \P\ suggests that for unary weights, one should be able to do it. Interestingly, we know very little about the case of weights encoded in binary: If the weights are encoded in binary, it is not known whether \SOC\ or \BCPM\ are in \RP\ or \NP-hard (since the randomized algorithm in  \cite{geelenComputingGirthCogirth2018} works only for unary weights). Note that the same questions were also asked by Schlotter and Seb{\H o} \cite{schlotterOddPathsCycles2025}.

While in the general case, derandomization of unary \SOC\ and \BCPM\ remains open, 
combining our algorithm for \EWPM\ with our reductions also implies that \BCPM, \SOC, and \ECS\ can be deterministically solved in time $n^{O(\ell)}$ if at most $\ell$ many different possible weights of perfect matchings 
(cycle sums, respectively) appear in the instance.

\subsection{Further Related Work on \EM}
\label{ssec:further_related_work}

As mentioned before, \EWPM\ with weights restricted to $0$ and $1$ is called the Exact Matching Problem (EM). Derandomizing the known randomized algorithm is an important open problem. Solution approaches have mainly been focused on three directions: Finding deterministic approximation algorithms, finding deterministic algorithms on restricted graph classes, and identifying other polynomial-time equivalent problems. 
\begin{itemize}
    \item In terms of approximation algorithms for \EM, Yuster~\cite{yusterAlmostExactMatchings2012} proved that one can find an almost exact matching (a matching of weight exactly $k$ that may fail to cover two vertices of the graph) in deterministic polynomial time if an exact matching exists in the graph. One could interpret this as an approximation result that relaxes the perfect matching constraint. If we relax the exactness constraint instead, recent progress includes a two-sided $2$-approximation algorithm on bipartite graphs by El~Maalouly~\cite{elmaaloulyExactMatchingAlgorithms2023}, and a one-sided $3$-approximation algorithm on bipartite graphs by Dürr, El Maalouly, and Wulf~\cite{durrApproximationAlgorithmExact2023}. 
    
    \item \EM\ has been studied on various restricted graph classes. In fact, Karzanov~\cite{karzanovMaximumMatchingGiven1987} proved already in 1987 that \EM\ can be solved in deterministic polynomial time on complete and complete bipartite graphs. This has been generalized in different ways: El Maalouly and Steiner~\cite{elmaaloulyExactMatchingGraphs2022} gave deterministic polynomial-time algorithms on graphs with bounded independence number (and bipartite graphs of bounded bipartite independence number). In the case of bipartite graphs, this was improved to an FPT-algorithm by El Maalouly, Steiner, and Wulf~\cite{elmaaloulyExactMatchingCorrect2023}. In the general case, Murakami and Yamaguchi~\cite{murakamiFPTAlgorithmExact2025} gave an FPT-algorithm that is paramatrized by the independence number and the minimum size of an odd cycle transversal. 
    
    El Maalouly, Haslebacher, and Wulf~\cite{elmaaloulyExactMatchingProblem2024} used Karzanov's original techniques to prove that \EM\ restricted to chain graphs, unit-interval graphs, and complete $r$-partite graphs is in \P. They also showed that a local search algorithm works in deterministic polynomial time on complete $r$-partite graphs, graphs of bounded neighborhood diversity, and graphs without complete bipartite $t$-holes.
    
    Finally, it is also possible to derandomize the randomized algorithm on some graph classes, e.g.\ by finding so-called Pfaffian orientations. Concretely, Vazirani~\cite{vaziraniNCAlgorithmsComputing1988} showed how to do this for $K_{3, 3}$-free graphs, and Galluccio and Loebl~\cite{galluccioTheoryPfaffianOrientations1999} obtained Pfaffian orientations for graphs of bounded 
    genus.

    \item Previously, some effort has also been made to identify other natural problems that are polynomial-time equivalent to \EM. We already mentioned some of these results in Section~\ref{ssec:related_cycle_problems}. 
    Besides the mentioned cycle problems, \EM\ is also known to be equivalent to a problem called Top-k Perfect Matching~\cite{elmaaloulyExactMatchingTopk2022}.

    \item Finally, Jia, Svensson, and Yuan~\cite{jiaExactBipartiteMatching2023} recently proved that the bipartite exact matching polytope has exponential extension complexity. This stands in contrast to the bipartite perfect matching polytope, which is known to be integral. 
\end{itemize}
We remark that many of these algorithmic results only hold for 0-1-weighted {\EWPM}, and do not even translate to unary \EWPM\ since they often involve very specific graph structures or parameters. In contrast, our parameterized algorithm works even for binary encoding of weights in \EWPM.

\subsection{Outline}
\label{ssec:outline}

We formally define \EWPM\ and related problems in Section~\ref{sec:preliminaries}, where we also recall some of the most important concepts from the literature. 
In Section~\ref{sec:l-best_matching}, we introduce the $\ell$-th Smallest Perfect Matching (\lBPM) problem, and we prove that it can be solved deterministically in time $n^{\bigO(\ell)}$. 
In Section~\ref{sec:related_cycle_problems}, we explain our reductions between cycle and matching problems. Full proofs of all statements in Section~\ref{sec:related_cycle_problems} can be found in Appendix~\ref{appendix:cycles_equiv_matchings}. 

\section{Preliminaries}
\label{sec:preliminaries}

We start by introducing notation and afterwards recall important facts from the literature. First of all, all our graphs are considered to be simple, 
i.e.\ they do not contain self-loops or parallel edges. If not specified, graphs 
are also considered to be undirected. Now, given a graph $G = (V, E)$,
we will frequently use the standard
notions of paths, cycles, and matchings, and we usually think of those 
objects as subsets of edges. In particular, 
a perfect matching $M$ of $G$ is a subset of edges $M \subseteq E$ such 
that every vertex of $G$ is covered exactly once by $M$. 
We denote by $\PMs_G$ the set of all perfect matchings of $G$.

Occasionally, we will explicitly use directed graphs. In particular,
we will write $(u, v) \in E$ (instead of $\{u, v\} \in E$) for directed edges in a 
directed graph $G = (V, E)$ with $u, v \in V$. 

In some problems, edge-weights are given for a graph $G = (V, E)$ in terms of a function 
$w : E \rightarrow \Z$ or $w : E \rightarrow \N$. We use the convention that $w(X)$
denotes the sum of all weights of edges in $X \subseteq E$. 

For two decision problems $A$ and $B$, we use $A \leq_p B$ to denote that there exists a polynomial-time many-one reduction from $A$ to $B$. We also write $A \equiv_p B$ to say that both $A \leq_p B$ and $B \leq_p A$ are true. 

\subsection{Symmetric Difference and Alternating Cycles}

An important and common tool for working with perfect matchings is 
the \emph{symmetric difference}. 
The symmetric difference of two sets $X, Y$ is defined as
$X \Delta Y \coloneqq (X \cup Y) \setminus (X \cap Y)$.
It is a well-known fact that the symmetric difference $M \Delta M'$ of two perfect matchings $M, M'$
in a graph $G = (V, E)$
always consists of a set of vertex-disjoint even cycles. 
In other words, there exist
vertex-disjoint even cycles $C_1, \dots, C_p \subseteq E$ such that $M \Delta M' = \bigsqcup_{i = 1}^p C_i$.
We say that $C_i$ is a cycle from $M \Delta M'$, and we sometimes use the letter $p$ to
refer to the number of cycles in $M \Delta M'$. Further, we recall that every cycle $C$ in $M \Delta M'$ is alternating in both $M$ and $M'$ (an even cycle $C$ of $G$ is called \emph{$M$-alternating} 
if and only if it alternates between edges in $M$ and edges not in $M$).
The set $M \Delta C$ is again a matching of the same cardinality as $M$.
We often refer to this property by saying that we \emph{switch} $M$ along $C$ to obtain the new matching $M \Delta C$. 
Naturally, the matchings $M \Delta C$ and $M$ differ exactly in $C$, i.e.\ $(M \Delta C) \Delta M = C$. 

Recall again that $M \Delta M'$
is a collection of vertex-disjoint even cycles $C_1, \dots, C_p \subseteq E$, and that each of the cycles $C_1, \dots, C_p$ is both $M$-alternating and $M'$-alternating. 
Hence, given both $M,M'$, some new perfect matchings of $G$ can be obtained by switching either $M$ or $M'$ along 
a subset of the cycles $C_1, \dots, C_p$. In particular, we can perform the switching operations
independently of each other because the cycles are vertex-disjoint.

\subsection{Minimum-Weight and Exact-Weight Perfect Matching}
Finding a minimum-weight perfect matching is a classical problem that has been studied extensively. We will define it as follows.

\begin{definition}[\MWPM]
    Given a graph $G = (V, E)$, edge-weights $w : E \rightarrow \Z$, 
    and an integer $k$, Minimum-Weight Perfect Matching (\MWPM) is the problem of deciding whether 
    $G$ admits a perfect matching of total weight at most $k$. 
\end{definition}
It is well-known that \MWPM\ is in \P\ (even if weights are encoded in binary, see e.g.\@~\cite{korteCombinatorialOptimizationTheory2018}). Moreover, finding a maximum-weight perfect matching 
is analogous to finding a minimum-weight perfect matching, and hence that can also be done in deterministic polynomial time. Since both of these problems are in \P, it then seems natural to ask whether we can also test for a perfect matching of weight exactly $k$ in polynomial time.

\begin{definition}[\EWPM]
    Given a graph $G = (V, E)$, edge-weights 
    $w : E \rightarrow \Z$, and an integer $k$,
    Exact-Weight Perfect Matching (\EWPM) is the problem
    of deciding whether $G$ admits a perfect matching of weight exactly $k$. 
\end{definition}
As mentioned in Section~\ref{sec:introduction}, the variant of \EWPM\ with weights from $\{0, 1\}$ is also known as Exact Matching (\EM), and it is contained in \RP~\cite{mulmuleyMatchingEasyMatrix1987}. Moreover, this also implies that \EWPM\ with weights encoded in unary is contained in \RP. In contrast, \EWPM\ with weights encoded in binary is \NP-complete~\cite{papadimitriouComplexityRestrictedSpanning1982, gurjarPlanarizingGadgetsPerfect2012}. Note that this constitutes a major difference to \MWPM, where both the unary and binary variant admit efficient algorithms.

The following problem \BCPM\ is at most as hard as \EWPM, and it was used in~\cite{elmaaloulyExactMatchingCorrect2023} as part of an algorithm for \EM.

\begin{definition}[\BCPM]
    Given a graph $G = (V, E)$, edge-weights 
    $w : E \rightarrow \Z$, and an integer $k$,
    Bounded Correct-Parity Perfect Matching (\BCPM) is the problem
    of deciding whether $G$ admits a perfect matching $M$ of weight $w(M)$ at most $k$ satisfying $k \equiv_2 w(M)$.
\end{definition}
The randomized algorithm for unary \EWPM\ also implies a randomized algorithm for unary \BCPM. Interestingly, the complexity of binary \BCPM\ remains open. In particular, it is not known whether binary \BCPM\ is \NP-hard or not.

\subsection{Conservative Weights and Cycle Problems}
\label{ssec:conservative_weights_and_cycle_problems}

Many problems on finding shortest paths and cycles in graphs with 
non-negative edge-weights can be solved efficiently. On the other hand, allowing 
for negative weights usually makes those problems \NP-hard as they now include 
problems such as finding Hamiltonian paths or cycles. An interesting compromise is provided 
by conservative edge-weights.

\begin{definition}[Conservative Weights]
    Edge-weights $w : E \rightarrow \Z$ of a (possibly directed) graph $G = (V, E)$ are 
    conservative if and only if there are no negative (directed) cycles in $G$ with respect to $w$. 
\end{definition}
Using established techniques (see e.g.\ Johnson's algorithm~\cite{johnsonEfficientAlgorithmsShortest1977} 
or Suurballe's algorithm~\cite{suurballeDisjointPathsNetwork1974}), 
conservative weights on directed graphs can be turned into non-negative weights 
while maintaining the total weight of each cycle. Hence, cycle problems in directed 
graphs with conservative 
edge-weights are usually not harder than their non-negative variants. 
However, in the case of undirected graphs, the same 
technique is not applicable, and it turns out that there are interesting 
cycle problems on undirected graphs with conservative edge-weights that relate to
\EWPM\ and \BCPM. Concretely, we will work with the two problems Exact Cycle Sum (\ECS) and Shortest Odd Cycle (\SOC), defined on graphs with conservative weights.

\begin{definition}[ECS]
\label{def:ecs}
    Given a graph $G = (V, E)$ with conservative 
    edge-weights $w : E \rightarrow \Z$ 
    and an integer $k$, Exact Cycle Sum (ECS) is the problem of deciding whether 
    there exists a set of vertex-disjoint cycles of total weight exactly $k$ in $G$.    
\end{definition}

\begin{definition}[SOC]
\label{def:soc}
    Given a graph $G = (V, E)$ with conservative 
    edge-weights $w : E \rightarrow \Z$ 
    and an integer $k$, Shortest Odd Cycle (SOC) is the problem of deciding whether 
    there exists a cycle of odd weight at most $k$ in $G$.
\end{definition}
Note that the definition of \SOC\ by Schlotter and Seb{\H o}~\cite{schlotterOddPathsCycles2025} as well as Geelen and Kapadia~\cite{geelenComputingGirthCogirth2018} slightly differs from ours: They require the cycle to have odd length (i.e.\ an odd number of edges) instead of odd weight. However, it is not hard to see that the two versions are polynomial-time equivalent. Indeed, by splitting even-weight edges, we can ensure to get a graph with only odd-weight edges, yielding a reduction from our odd-weight version to the odd-length version. To reduce the odd-length version to the odd-weight version, it suffices to manipulate the weights, replacing the weight $w(e)$ of each edge $e$ with $2|V| w(e) + 1$.

As mentioned in Section~\ref{ssec:related_cycle_problems}, El Maalouly et al.\@~\cite{elmaaloulyExactMatchingCorrect2023} gave a polynomial-time 
algorithm for \BCPM\ on bipartite graphs by reducing it to a directed version of \SOC. An analogous connection between \EWPM\ on bipartite graphs and the directed version of \ECS\ has been observed by Papadimitriou and 
Yannakakis~\cite{papadimitriouComplexityRestrictedSpanning1982} already in 1982. In Section~\ref{sec:related_cycle_problems}, we will 
extend these connections between cycle and matching problems by proving equivalence of \ECS\ (undirected) and \EWPM\ (on general graphs), and \SOC\ (undirected) and \BCPM\ (on general graphs), respectively.

\section{\texorpdfstring{$\ell$}{l}-th Smallest Perfect Matchings}
\label{sec:l-best_matching}
 
As discussed in Section~\ref{sec:introduction}, the problem of finding a deterministic polynomial-time algorithm for unary \EWPM\ has been attacked from various directions in recent years. We attack it from yet another direction that involves so-called $\ell$-th smallest perfect matchings.

\begin{definition}
    Let $G = (V, E)$ be a graph with edge-weights $w : E \rightarrow \Z$. We define $\ell$-th smallest perfect matchings for all $\ell \in \N$ recursively. The $1$st-smallest perfect matchings of $G$ are exactly the minimum-weight perfect matchings. Moreover, for fixed $\ell \geq 2$, a perfect matching $M$ is $\ell$-th smallest if and only if it has minimum weight among all perfect matchings that are not $\ell'$-th smallest for any $\ell' < \ell$.
\end{definition}
\begin{definition}[\lBPM]
    Given a graph $G = (V, E)$, edge-weights 
    $w : E \rightarrow \Z$, an integer $k$, and a natural number $\ell$,
    $\ell$-th Smallest Perfect Matching (\lBPM) is the problem
    of deciding whether $G$ admits an $\ell$-th smallest perfect matching of weight exactly $k$. 
\end{definition}
The goal of this section is to give a deterministic algorithm that can decide \lBPM\ in time $n^{\bigO(\ell)}$, where $n$ is the size of the instance. Naturally, this has interesting consequences for \EWPM: Assume that we are given an instance of \EWPM\ where $k$ is such that any perfect matching $M$ with $w(M) = k$ is an $\ell'$-th smallest perfect matching. Then we can verify existence of such a matching in time $n^{\bigO(\ell')}$ by simply calling the \lBPM-algorithm repeatedly for $\ell = 1, \dots, \ell'$. 

To the best of our knowledge, we are the first to observe a result of this kind. Our algorithm is very simple and follows a natural approach: For a given $\ell$ and for every set $F \subseteq E$ of size at most $2 (\ell - 1)$, compute a minimum-weight perfect matching that includes all of $F$. We prove in Theorem~\ref{theorem:l-best} that there is always a $1$st-smallest, a $2$nd-smallest, \dots, an $\ell$-th smallest perfect matching among the set of perfect matchings that were computed. The proof of Theorem~\ref{theorem:l-best} is less straight-forward. It relies on the complementary slackness theorem (compare e.g.\ \cite{korteCombinatorialOptimizationTheory2018}) and the perfect matching polytope.

\begin{theorem}[$\ell$-th Smallest PM in General Graphs]
\label{theorem:l-best}
    Let $G = (V, E)$ be a graph with edge-weights $w : E \rightarrow \Z$. Let $\ell \in \N$ be arbitrary. Assume that $G$ contains an $\ell$-th smallest perfect matching $M$. Then there exists a subset of edges $F \subseteq E$ of size at most $2(\ell - 1)$ such that 
    \begin{displaymath}
        w(M) =  \min\limits_{M' \in \PMs_G \, : \, F \subseteq M'} w(M').
    \end{displaymath}
\end{theorem}

It turns out that one can slightly improve the above theorem in the case of bipartite graphs. 

\begin{theorem}[$\ell$-th Smallest PM in Bipartite Graphs]
\label{theorem:l-best-bipartite}
    The $\ell$-th smallest perfect matching in a weighted bipartite graph can be obtained the same way as in Theorem~\ref{theorem:l-best}, except that $F$ ranges only over edge sets of size at most $\ell - 1$.
\end{theorem}

In this paper, we omit the proof of Theorem~\ref{theorem:l-best-bipartite}, since its proof is identical to that of Theorem~\ref{theorem:l-best}, except that the general matching polytope is swapped with the bipartite matching polytope.

One might ask whether the number of edges one needs to "fix" in these two theorems is tight. 
Indeed, this is the case, i.e.\ in order to find the $\ell$-th smallest perfect matching, 
one needs to fix $\ell - 1$ edges in a bipartite graph and $2(\ell - 1)$ edges in a general graph. 
To see this, consider Figure~\ref{figure:tight_bounds}, where dashed edges have a weight of $1$ while solid edges have a weight of $0$. 
Consider first the bipartite four-cycle on the left. It has exactly two perfect matchings, one of weight $0$ and one of weight $2$. 
Fixing one edge is necessary and sufficient to find the $2$nd-smallest perfect matching. 
Taking multiple disjoint copies of this four-cycle generalizes this to show that it is sometimes necessary to fix $\ell - 1$ edges to find the $\ell$-th smallest perfect matching in a bipartite graph. 
Analogously, the graph on the right has perfect matchings of weight $1$ and $3$, but fixing a single edge does not suffice to find the perfect matching of weight $3$. 
In other words, we need to fix $2$ edges to find the second-smallest perfect matching here. 
Again, generalizing this argument by taking multiple copies of this graph proves that the bound $2(\ell - 1)$ is best possible in general graphs. 
\begin{figure}[htb]
    \centering
    \includegraphics[width=0.6\textwidth]{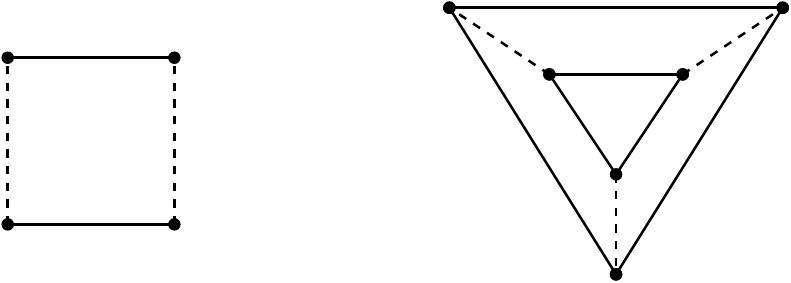} 
    \caption{Example that Theorems~\ref{theorem:l-best} and \ref{theorem:l-best-bipartite} are tight in terms of the number of fixed edges. Dashed edges have a weight of $1$ and solid edges have a weight of 0.}
    \label{figure:tight_bounds}
\end{figure}

\subsection{Proof of Theorem~\ref{theorem:l-best}}
\label{subsection:general_graphs}

The goal of this section is to prove Theorem~\ref{theorem:l-best}. We start by recalling some useful facts about the perfect matching polytope from Chapter~11 of the standard textbook by Korte and Vygen~\cite{korteCombinatorialOptimizationTheory2018}.
For this, it will be convenient to use the notation $\delta(v)$ to denote the set of edges incident to a vertex $v \in V$. Similarly, for all $A \subseteq V$, we use $\delta(A)$ to denote the set of edges with one endpoint in $A$ and the other in $V \setminus A$. Moreover, we denote by $\A$ the set of all odd-cardinality sets $A \subseteq V$, i.e.\ $\A \coloneqq \{ A \subseteq V : |A| \equiv_2 1\}$.

For the remainder of this section, let $G = (V, E)$ be a graph with edge-weights $w : E \rightarrow \Z$. 
The integer linear program 
\begin{align*}
    &\min \sum_{e \in E} x_e w(e) \\
    \text{ s.t.\ } &\sum_{e \in \delta(v)} x_e = 1 &\forall v \in V\\
    \text{ and } &x_e \in \{0, 1\} &\forall e \in E
\end{align*}
with variables $\{x_e \mid e \in E \}$ captures the problem of finding a minimum-weight perfect matching in $G$. Now consider the following LP relaxation
\begin{align*}
    &\min \sum_{e \in E} x_e w(e) \\
    \text{ s.t.\ } &\sum_{e \in \delta(v)} x_e = 1 &\forall v \in V\\
    &\sum_{e \in \delta(A)} x_e \geq 1 &\forall A \in \A \text{ with } |A| > 1 \\
     &x_e \geq 0 &\forall e \in E
\end{align*}
where additional odd-set constraints were added. By the famous result of Edmonds~\cite{edmonds1965maximum}, adding these additional odd-set constraints is necessary and sufficient to make the LP integral.
In other words, the basic feasible solutions to the LP are exactly the perfect matchings of $G$. The dual LP is now given by \begin{align*}
    &\max \sum_{A \in \A} z_A \\
    \text{ s.t.\ } &\sum_{A \in \A \, : \,  e \in \delta(A)} z_A \leq w(e) &\forall e \in E\\
     & z_A \geq 0 &\forall A \in \A \text{ with } |A| > 1\\
\end{align*}
where we introduced a variable $z_A$ for all $A \in \A$. By strong duality, the optimal value of the primal and the dual coincide.

\begin{definition}[Residual Graph]
    Let $G = (V, E)$ be a graph with edge-weights $w : E \rightarrow \Z$. 
    Let $z$ be an optimal solution to the dual LP associated with $G$.
    The graph $G_z \coloneqq (V, E_z)$ with 
    \[
        e \in E_z \iff w(e) = \sum_{A \in \A \, : \, e \in \delta(A)} z_A
    \]
    is called residual graph with respect to $z$.
\end{definition}
Unsurprisingly, it turns out that all minimum-weight perfect matchings are contained in the residual graph with respect to a given optimal solution $z$.
However, there might exist perfect matchings in the residual graph that are not minimum-weight perfect matchings of $G$. Hence, we need an extra condition to exactly characterize which perfect matchings in the residual graph are minimum-weight in $G$. 
Eppstein \cite{eppstein1995representing} shows such a result for bipartite graphs.
We present here the natural generalization to all graphs. This essentially follows from the complementary slackness theorem \cite{korteCombinatorialOptimizationTheory2018}.

\begin{lemma}
\label{lemma:general_residual_graph}
    Let $G = (V, E)$ be a graph with edge-weights $w : E \rightarrow \Z$. 
    Let $z$ be an optimal solution to the dual LP associated with $G$, and let $G_z = (V, E_z)$ be its residual graph.
    A perfect matching $M$ of $G$ has minimum weight if and only if $M \subseteq E_z$ and $|M \cap \delta(A) | = 1$ holds for all $A \in \A$ with $z_A > 0$.
\end{lemma}
\begin{proof}
For the sake of clarity, we prove both directions individually.
\begin{itemize}
    \item[($\Rightarrow$)] We provide an indirect proof. Assume first that there exists $f \in M$ such that 
    \[
        \sum_{A \in \A \, : \,  f \in \delta(A)} z_A < w(f)
    \]
    holds and hence $f \notin E_z$. Observe that, by feasibility of $z$, we have 
    \[
        \sum_{A \in \A \, : \,  e \in \delta(A)} z_A \leq w(e)
    \]
    for all $e \in E$. Moreover, we have $z_A \geq 0$ for all $A \in \A$ with $|A| > 1$. Thus, we get 
    \[       
        w(M) - \sum_{A \in \A} z_A \geq \sum_{e \in M} \left( w(e) - \sum_{A \in \A \, : \, e \in \delta(A)} z_A \right) > 0
    \]
    since every set $A \in \A$ with $|A| = 1$ is associated with exactly one edge in $M$, and every other set $A \in \A$ (with $|A| > 1$) satisfies $z_A \geq 0$. In other words, we have 
    \[
        w(M) > \sum_{A \in \A} z_A
    \]
    which, by strong duality, implies that $M$ does not have minimum weight. Thus, assume now that $M \subseteq E_z$, but there exists $A \in \A$ such that $z_A > 0$ and $|M \cap A| > 1$. Again, we calculate 
    \[       
        w(M) - \sum_{A \in \A} z_A > \sum_{e \in M} \left( w(e) - \sum_{A \in \A \, : \, e \in \delta(A)} z_A \right) = 0
    \]
    where the strict inequality now comes from the fact that we needed to duplicate the variable $z_A$ in the first step. Again, by strong duality, we conclude that $M$ does not have minimum weight.

    \item[($\Leftarrow$)] Assume now that we have $M \subseteq E_z$ and $|M \cap A| = 1$ for all $A \in \A$ with $z_A > 0$. In this case, we get
    \[
        w(M) - \sum_{A \in \A} z_A = \sum_{e \in M} \left( w(e) - \sum_{A \in \A \, : \, e \in \delta(A)} z_A \right) = 0
    \]
    as every set $A \in \A$ with $z_A > 0$ is associated with exactly one edge in $M$. By strong duality, we can use the fact that $z$ is optimal to conclude that $M$ has minimum weight.
\end{itemize}
\end{proof}
Using Lemma~\ref{lemma:general_residual_graph}, we can prove the following lemma, which captures the special case $\ell = 2$ of Theorem~\ref{theorem:l-best}.

\begin{lemma}
\label{lemma:fixing_two_edges}
    Let $G = (V, E)$ be a graph with edge-weights $w : E \rightarrow \Z$. Assume that $G$ contains a $2$nd-smallest perfect matching $M$. Then there exist edges $e, f \in E$ such that 
    \begin{displaymath}
        w(M) =  \min\limits_{M' \in \PMs_G \, : \, \{e, f\} \subseteq M'} w(M').
    \end{displaymath}
\end{lemma}
\begin{proof}
    Let $z$ be a fixed optimal dual solution. Since $M$ is not a minimum-weight perfect matching, by Lemma~\ref{lemma:general_residual_graph}, there are two cases: Either there exists an edge $e \in M$ with $e \not\in E_z$ (in this case, take as $f$ any edge from $M$ and we are done), or there exists some odd set $A \in \mathcal{A}$ with $z_A > 0$, distinct edges $e, f \in M$, and both $e,f \in \delta(A)$. In the latter case, by the same lemma, no minimum-weight perfect matching contains both $e$ and $f$. Hence, we get
    \begin{displaymath}
        \text{OPT} < \min\limits_{M' \in \PMs_G \, : \, \{e, f\} \subseteq M'} w(M') = w(M)
    \end{displaymath}
    where OPT denotes the weight of a minimum-weight perfect matching in $G$.
\end{proof}
It remains to use Lemma~\ref{lemma:fixing_two_edges} repeatedly to obtain Theorem~\ref{theorem:l-best}.

\begin{proof}[Proof of Theorem~\ref{theorem:l-best}]
    We proceed by induction over $\ell$. The case $\ell = 1$ is trivial and Lemma~\ref{lemma:fixing_two_edges} already proves the case $\ell = 2$. Thus, let now $\ell \geq 3$ be arbitrary, and assume as induction hypothesis that the statement holds for all $\ell' < \ell$. Recall that $M$ is an $\ell$-th smallest perfect matching. We want to prove that there is a subset $F \subseteq E$ of edges of size at most $2(\ell - 1)$ such that 
    \begin{displaymath}
        w(M) =  \min\limits_{M' \in \PMs_G \, : \, F \subseteq M'} w(M').
    \end{displaymath}
    Let $M_{\min} $ be an arbitrary minimum-weight perfect matching of $G$. From the argument in Lemma~\ref{lemma:fixing_two_edges}, we know that there exist edges $e, f \in M$ such that 
    \begin{displaymath}
        w(M_{\min}) < \min\limits_{M' \in \PMs_G \, : \, \{e, f\} \subseteq M'} w(M') \leq w(M).
    \end{displaymath}
    Consider now the graph $G'$ obtained from $G$ by removing $e$ and $f$, their incident vertices, and all their incident edges. 
    Observe that $M \setminus \{e, f\}$ is an $\ell'$-th smallest perfect matching in $G'$ for some $\ell' < \ell$. By the induction hypothesis, this means that there is a set $F' \subseteq E \setminus \{e, f\}$ of size at most $2(\ell' - 1)$ such that 
    \begin{displaymath}
        w(M \setminus \{e, f\}) =  \min\limits_{M' \in \PMs_{G'} \, : \, F' \subseteq M'} w(M').
    \end{displaymath}
    Define $F \coloneqq F' \cup \{e, f\}$ and observe that 
    \begin{displaymath}
        w(M) = w(\{e, f\}) + \min\limits_{M' \in \PMs_{G'} \, : \, F' \subseteq M'} w(M') = \min\limits_{M' \in \PMs_{G} \, : \, F \subseteq M'} w(M')
    \end{displaymath}
    and that the set $F$ has size at most $2(\ell - 1)$.
\end{proof}

\section{Related Cycle Problems}
\label{sec:related_cycle_problems}

In this section, we present the ideas behind Theorem~\ref{theorem:cycles_equiv_matchings} below. The full proof of Theorem~\ref{theorem:cycles_equiv_matchings} and the involved lemmata can be found in Appendix~\ref{appendix:cycles_equiv_matchings}. 

\begin{restatable}[\EWPM\ $\equiv_p$ \ECS\ \& \BCPM\ $\equiv_p$ \SOC]{theorem}{theoremcyclesequivmatchings}
\label{theorem:cycles_equiv_matchings}
    \EWPM\ is polynomial-time equivalent to \ECS. 
    \BCPM\ is polynomial-time equivalent to \SOC. 
\end{restatable}

As we will see, both equivalences use the same base construction and only differ in some minor configurations. The encoding of the weights is preserved in all reductions and hence \ECS\ is equivalent to \EM\ for weights encoded in unary, and \NP-complete for weights encoded in binary.

\paragraph*{From Matchings to Cycles}

We start with the direction of \EWPM\ $\leq_p$ \ECS\ and \BCPM\ $\leq_p$ {\SOC}, i.e.\ we give reductions starting from matching-type problems going to cycle-type problems.
We use roughly the same construction for both reductions. We remark that in the case of matching-type problems in \emph{bipartite graphs}, a similar reduction already exists, and it goes to cycle-type problems in \emph{directed} graphs \cite{papadimitriouComplexityRestrictedSpanning1982, elmaaloulyExactMatchingCorrect2023}. 
In these existing reductions,
the bipartition of the graph was used to define some edge orientation, which in turn ensured that cycles in the new graph would correspond to $M$-alternating cycles (where $M$ is an arbitrary minimum-weight perfect matching) in the old graph (see \cite{papadimitriouComplexityRestrictedSpanning1982} and \cite{elmaaloulyExactMatchingCorrect2023} for more details). 
The same trick does not apply in general graphs, which is why we need to reduce to undirected cycle problems instead. In particular, the new trick will be to manipulate
the weights such that the weight of any non-alternating cycle exceeds some threshold. 
The following lemma captures the key properties of this construction. 

\begin{restatable}{lemma}{lemmamatchingstocycles}
\label{lemma:matchings_to_cycles}
    Let $G = (V, E)$ be a graph with non-negative edge-weights 
    $w : E \rightarrow \N$ that admits a minimum-weight perfect matching $M$. 
    Consider the new edge-weights $w' : E \rightarrow \Z$ defined as
    \[
        w'(e) \coloneqq \begin{cases}
                            - w(e) - r - 1, & e \in M \\
                            w(e) + r + 1, & e \notin M
                        \end{cases}  
    \]
    for some integer $r \geq w(M)$. 
    Then every cycle $C$ with $w'(C) \leq r - w(M)$ is $M$-alternating
    and $w'$ is conservative.
\end{restatable}

Using Lemma~\ref{lemma:matchings_to_cycles}, we can now informally describe the reductions: Assume first that we are given an \EWPM-instance $I$ consisting of $G = (V, E)$, $w : E \rightarrow \Z$, and $k \in \Z$. 
Without loss of 
generality, we can assume that the weights $w$ are non-negative. We first compute 
a minimum-weight perfect matching $M$ in $G$. Note that if no such 
matching exists or its weight exceeds $k$, $I$ is a NO-instance and we can produce a 
simple NO-instance of \ECS. Otherwise, we construct $w'$ as in 
Lemma~\ref{lemma:matchings_to_cycles} (setting $r = k$) and return the \ECS-instance 
$I'$ consisting of $G$, $w'$ and $k' \coloneqq k - w(M)$. By Lemma~\ref{lemma:matchings_to_cycles}, $w'$ is 
conservative. Setting $r = k$ crucially guarantees that solutions to the \ECS-instance can be translated back to solutions to the \EWPM-instance $I$. 

The only difference in the reduction from \BCPM\ to \SOC\ is that we need to manually check the parity of $M$ and use a trivial construction if it already has correct parity.

\paragraph*{From Cycles to Matchings}
We will now turn our attention to the reverse reductions \ECS\ $\leq_p$ \EWPM\ and \SOC\ $\leq_p$ \BCPM. The general idea is again inspired by the existing reductions between bipartite matching and directed cycle problems (\cite{papadimitriouComplexityRestrictedSpanning1982, elmaaloulyExactMatchingCorrect2023}):
We want to 
duplicate every vertex and add edges between duplicates in order to ensure the existence of a canonical perfect matching $M$.
Additionally, we want to ensure that the existence of any other perfect matching $M'$ in the new graph implies existence of cycles in the original 
graph by looking at the symmetric difference $M \Delta M'$. Unfortunately, since our original graph
is not directed anymore, we cannot use the ideas from~\cite{papadimitriouComplexityRestrictedSpanning1982, elmaaloulyExactMatchingCorrect2023} to achieve this. Instead, we use the construction described in Lemma~\ref{lemma:cycles_to_matchings} and Figure~\ref{figure:cycle_to_matching}. 

\begin{figure}[htb]
    \centering
    \includegraphics[width=0.7\textwidth]{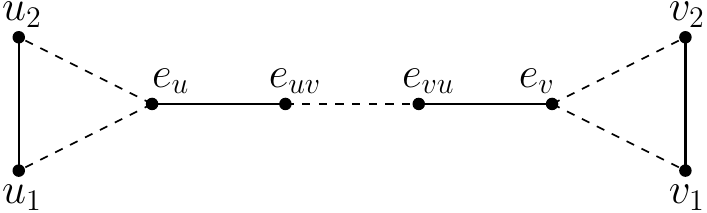} 
    \caption{This figure shows the construction from Lemma~\ref{lemma:cycles_to_matchings}
    for a single edge $e = \{u, v\}$. Solid edges indicate the perfect matching 
    $M$ from Lemma~\ref{lemma:cycles_to_matchings}. The edge $\{e_{uv}, e_{vu}\}$
    receives the weight of $e$ while all other edges get a weight of zero.}
    \label{figure:cycle_to_matching}
\end{figure}

\begin{restatable}{lemma}{lemmacyclestomatchings}
\label{lemma:cycles_to_matchings}
    Let $G = (V, E)$ be a graph with conservative edge-weights $w : E \rightarrow \Z$.
    We define
    \begin{align*}
        V' \coloneqq &\{v_1, v_2 \, | \, v \in V\} \cup \{e_u, e_{uv}, e_{vu}, e_v \, | \, e = \{u, v\} \in E\} \\
        E' \coloneqq &\{ \{v_1, v_2\} \, | \, v \in V\} \\
        &\cup \{\{u_1, e_u\}, \{u_2, e_u\}, \{e_u, e_{uv}\},  \{e_{uv}, e_{vu}\}, \\ 
        &\qquad \{e_{vu}, e_v\}, \{e_v, v_1\}, \{e_v, v_2\} \, | \, e = \{u, v\} \in E\} \\
        w' \coloneqq &e' \in E' \mapsto \begin{cases}
            w(e) & \text{if } e' = \{e_{uv}, e_{vu}\} \text{ for some } e = \{u, v\} \in E \\
            0 & \text{otherwise}
            \end{cases}  \\
        M \coloneqq &\{ \{v_1, v_2\} \, | \, v \in V\} \cup \{\{e_u, e_{uv}\}, \{e_{vu}, e_v\} \, | \, e = \{u, v\} \in E\}
    \end{align*}
    and consider the graph $G' = (V', E')$ with edge-weights $w'$ and perfect matching 
    $M$. Then there exist vertex-disjoint cycles $C_1, \dots, C_p$ of total weight 
    $k$ (w.r.t.\ $w$) in $G$ if and only if there exist vertex-disjoint $M$-alternating 
    cycles $C'_1, \dots, C'_p$ of total weight $k$ (w.r.t.\ $w'$) in $G'$.
\end{restatable}

Let us now briefly explain how this yields the desired reductions. Assume that we are given an \ECS-instance $I$ consisting of $G = (V, E)$, $w : E \rightarrow \Z$, and $k \in \Z$. We use the construction from Lemma~\ref{lemma:cycles_to_matchings} to obtain a graph $G'$ with edge-weights $w'$ and return the \EWPM-instance $I'$ consisting of $G' = (V', E')$, $w'$, and $k' \coloneqq k$. The symmetric difference $M \Delta M'$ of $M$ with any perfect matching $M'$ in $G'$ yields a set of vertex disjoint cycles that can be translated back to $G$ using Lemma~\ref{lemma:cycles_to_matchings}. The reduction from \SOC\ to \BCPM\ uses the same construction, but we need to additionally observe that we can get a single odd-weight cycle that translates through Lemma~\ref{lemma:cycles_to_matchings}.

\section{Final Remarks}
\label{sec:final_remarks}

In Section~3, we proved that we can deterministically find an $\ell$-th smallest perfect matching in time $n^{\bigO(\ell)}$. This also allows us to solve certain instances of \EWPM\ efficiently: Assume e.g.\ that we are given an \EWPM\ instance with $k$ such that $k_1 < k_2 <  \dots < k_\ell = k$ are all the distinct weights up to $k$ that perfect matchings in $G$ can take. Then our algorithm can find a perfect matching of weight $k = k_\ell$ efficiently if $\ell$ is small. 

Clearly, this also works in the case of \BCPM. Moreover, our reductions from Section~\ref{sec:related_cycle_problems} also allow us to get similar guarantees for \ECS\ and \SOC, respectively. Indeed, assume that we are given an instance of \ECS\ or \SOC\ with $k$ such that $k_1 < k_2 <  \dots < k_\ell = k$ are all the distinct weights up to $k$ that cycle sums in $G$ can take. We apply our reduction to build a corresponding \EWPM\ or \BCPM\ instance. By Lemma~\ref{lemma:cycles_to_matchings}, our reductions have a one-to-one correspondence between $\ell$-th smallest perfect matchings and $\ell$-th smallest cycle sums. Hence, the parameterized runtime guarantees of the \lBPM-algorithm carry through the reduction. 
\begin{corollary}
\label{corollary:parametrized_algo_+_reduction}
    \BCPM\ can be solved in time $n^{\bigO(\ell)}$ if $k_1 < k_2 <  \dots < k_\ell = k$ are all the distinct weights up to $k$ that perfect matchings in $G$ can take. \ECS\ and \SOC\ can be solved in time $n^{\bigO(\ell)}$ if $k_1 < k_2 <  \dots < k_\ell = k$ are all the distinct weights up to $k$ that cycle sums in $G$ can take.
\end{corollary}



\bibliography{references}

\appendix

\section{Proofs of Section~\ref{sec:related_cycle_problems}}
\label{appendix:cycles_equiv_matchings}

This section contains the proofs that were omitted in Section~\ref{sec:related_cycle_problems}. Recall that the overall goal of Section~\ref{sec:related_cycle_problems} is to prove the following theorem.

\theoremcyclesequivmatchings*

Section~\ref{sec:related_cycle_problems} contains an informal description of the proof. Here, we prove both directions in detail. Concretely, we prove $\leq_p$ in Section~\ref{appendix:matchings_to_cycles} and $\geq_p$ in Section~\ref{appendix:cycles_to_matchings}.

\subsection{Matchings to Cycles}
\label{appendix:matchings_to_cycles}

Recall that the main tool for the $\leq_p$-direction is the following lemma. 

\lemmamatchingstocycles*
\begin{proof}
    We start with an indirect proof for the first part of the statement: Let $C$ 
    be a cycle in $G$ that is not $M$-alternating. Since $M$ is a matching, we 
    must have $\abs{C \cap M} < \abs{C \setminus M}$. Thus, we get a lower bound 
    of
    \begin{align*}
        w'(C) &= w'(C \cap M) + w'(C \setminus M) \\
        &\geq -w(M) - \abs{C \cap M} (r + 1)    
        + \abs{C \setminus M} (r + 1) \\
        &\geq r + 1 - w(M)
    \end{align*}
    for the weight of $C$. This proves the first part of the statement.

    For the second part, assume for a contradiction that there is a cycle $C$ in 
    $G$ of negative weight $w'(C) < 0$. By the above argument, $C$ must be $M$-alternating
    and in particular, we have  $\abs{C \cap M} = \abs{C \setminus M}$.
    Hence, $M' \coloneqq M \Delta C$ is a perfect matching of weight 
    \[
        w(M') = w(M) - w(C \cap M) + w(C \setminus M) = w(M) + w'(C) < w(M)
    \]
    which contradicts the minimality of $M$. We conclude that $w'$ is conservative.
\end{proof}

With Lemma~\ref{lemma:matchings_to_cycles}, we are now ready to prove
\EWPM\ $\leq_p$ \ECS\ and \BCPM\ $\leq_p$ \SOC. For completeness sake, we repeat the description of the construction from Section~\ref{sec:related_cycle_problems} in the proof.

\begin{proof}[Proof of $\leq_p$ in Theorem~18]
We start by proving \EWPM\ $\leq_p$ \ECS. The proof of \BCPM\ $\leq_p$ \SOC\ will be similar. 

Assume we are given an \EWPM-instance $I$ consisting of $G = (V, E)$, $w : E \rightarrow \Z$, and $k \in \Z$. 
Without loss of 
generality, we can assume that the weights $w$ are non-negative. We first compute 
a minimum-weight perfect matching $M$ in $G$. Note that if no such 
matching exists or its weight exceeds $k$, $I$ is a NO-instance and we can produce a 
simple NO-instance of \ECS. Otherwise, we construct $w'$ as in 
Lemma~\ref{lemma:matchings_to_cycles} (setting $r = k$) and return the \ECS-instance 
$I'$ consisting of $G$, $w'$ and $k' \coloneqq k - w(M)$. By Lemma~\ref{lemma:matchings_to_cycles}, $w'$ is 
conservative. Moreover, all of this can be done in polynomial time and hence 
it remains to argue correctness. 

To this end, assume first that $I$ is a 
YES-instance, i.e.\ there is a perfect matching $M^*$ in $G$ with 
$w(M^*) = k$. The symmetric difference $M \Delta M^*$ consists of 
vertex-disjoint cycles $C_1, \dots, C_p$. Moreover, by definition of $w'$,
we have $k = w(M^*) = w(M) + \sum_{i = 1}^p w'(C_i)$ and hence 
$\sum_{i = 1}^p w'(C_i) = k - w(M)$. Thus, $I'$ is a YES-instance.

Assume now for the other direction that $I'$ is a YES-instance, i.e.\ there 
exist cycles $C_1, \dots, C_p$ with $\sum_{i = 1}^p w'(C_i) = k - w(M)$.
Since $w'$ is conservative, the weight of each individual cycle is non-negative and 
in particular, we have $0 \leq w'(C_i) \leq k - w(M)$ for all $i \in [p]$. 
Hence, the cycles $C_1, \dots, C_p$ are all $M$-alternating by 
Lemma~\ref{lemma:matchings_to_cycles}. This again implies that 
$M^* \coloneqq M \Delta C_1 \Delta \dots \Delta C_p$ is a perfect matching 
in $G$. Observing that we have $w(M^*) = w(M) + \sum_{i = 1}^\ell w'(C_i)$, we 
conclude that $I$ is a YES-instance.

This concludes the proof of \EWPM\ $\leq_p$ \ECS, and we will move on to prove \BCPM\ $\leq_p$ \SOC. The reduction is very similar: Given a \BCPM-instance $I$ consisting of $G = (V, E)$, $w : E \rightarrow \Z$, and $k \in \Z$, 
we first compute a minimum-weight perfect matching $M$ and check $w(M) \equiv_2 k$.
If $w(M)$ indeed has the same parity as $k$, we can decide the instance by 
checking $w(M) \leq k$ and for the sake of the reduction, produce a simple 
YES-instance or NO-instance accordingly. Otherwise, we construct 
$I'$ as in the reduction \EWPM\ $\leq_p$ \ECS\ above. As all of this 
is doable in polynomial-time, it remains to argue correctness. 

For this, 
assume first that $I$ is a YES-instance, i.e.\ there is a perfect 
matching $M^*$ in $G$ with $w(M^*) \leq k$ and $w(M^*) \equiv_2 k$. 
It is clear that the construction is correct if $M$ has correct parity, so assume 
now $w(M) \nequiv_2 k$. The symmetric difference $M \Delta M^*$ consists 
of vertex-disjoint cycles $C_1, \dots, C_p$ of odd total weight 
$\sum_{i = 1}^p w'(C_i) = w(M^*) - w(M) \leq k - w(M)$. 
Hence, at least one of the cycles is 
of odd weight, let it be $C_i$. Recall that by Lemma~\ref{lemma:matchings_to_cycles},
$w'$ is conservative and hence we have $0 \leq w'(C_i) \leq k - w(M)$. We have 
established $1 \equiv_2 w'(C_i) \leq k - w(M)$ and hence $I'$ is a YES-instance. 

Assume now for the other direction that $I'$ is a YES-instance, i.e.\
there exists a cycle $C$ in $G$ with $1 \equiv_2 w'(C) \leq k - w(M)$. Recall 
that we assume $w(M) \nequiv_2 k$ as otherwise, the reduction clearly works out.
By Lemma~\ref{lemma:matchings_to_cycles}, $C$ must be $M$-alternating and 
hence $M^* \coloneqq M \Delta C$ is a perfect matching. Moreover, we have 
$w(M^*) = w(M) + w'(C) \leq k$ as well as $w(M^*) = w(M) + w'(C) \equiv_2 w(M) + 1 \equiv_2 k$.
Hence, $I$ is a YES-instance, as desired.
\end{proof}

\subsection{Cycles to Matchings}
\label{appendix:cycles_to_matchings}

Recall that the main tool for the $\geq_p$-direction is the following lemma.

\lemmacyclestomatchings*
\begin{proof}
    We prove the two directions separately. First, assume that there exist 
    vertex-disjoint cycles $C_1, \dots, C_p$ of weight $k$ in $G$. For convenience, 
    we orient each of the cycles arbitrarily. Now we define 
    \[
        C'_i \coloneqq \bigcup_{e = (u, v) \in C_i} \{ \{u_1, e_u\}, \{e_u, e_{uv}\}, 
        \{e_{uv}, e_{vu}\}, \{e_{vu}, e_v\}, \{e_v, v_2\}, \{v_2, v_1\}\}
    \]
    for all $i \in [p]$. By definition, the cycles $C'_1, \dots, C'_p$ are 
    $M$-alternating. Moreover, by construction of $w'$, we immediately have $w'(C'_i) = w(C_i)$
    and hence $C'_1, \dots, C'_p$ have total weight $k$. 
    It is left to prove that  $C'_1, \dots, C'_p$ are vertex-disjoint.
    Assume first to the contrary that 
    two distinct cycles $C'_i$ and $C'_j$ share a vertex $x \in V'$. If $x = v_1$
    or $x = v_2$ for some $v \in V$, then both $C_i$ and $C_j$ cover $v$ and hence
    they intersect. The alternative is that $x \in \{e_u, e_{uv}, e_{vu}, e_v\}$ for 
    some $e = \{u, v\} \in E$. But then both $C_i$ and $C_j$ must contain $e$ and 
    hence intersect. In both cases we come to a contradiction to the assumption that
    $C_i$ and $C_j$ are vertex-disjoint and we conclude that $C'_i$ and $C'_j$ must 
    be vertex-disjoint too.

    Assume now for the other direction that there exist vertex-disjoint 
    $M$-alternating cycles $C'_1, \dots, C'_p$ of weight $k$ in $G'$. We define 
    \[
        C_i \coloneqq \{ e = \{v, u\} \in E \, | \, \{e_{uv}, e_{vu}\} \in C'_i \}
    \]
    for all $i \in [p]$ and claim that $C_1, \dots, C_p$ satisfy the desired 
    properties. Again, $w(C_i) = w'(C'_i)$ follows directly from the definition of $w'$
    and $C_i$. Thus, it remains to prove that $C_1, \dots, C_p$ are vertex-disjoint.
    For this, we first observe the following property about the $M$-alternating 
    cycles $C'_1, \dots, C'_p$: For $v \in V$ and $i \in [p]$, we have 
    $v_1 \in C'_i$ if and only if $v_2 \in C'_i$. This property follows directly
    from the assumption that $C'_i$ is $M$-alternating and $M$ contains the edge $\{v_1, v_2\}$
    for all $v \in V$.

    Assume now for a contradiction that there exist distinct cycles 
    $C_i$ and $C_j$ that intersect 
    in vertex $v \in V$. Then there must exist edges $e = \{v, u\}$ and 
    $f = \{v, w\}$ in $E$ such that $\{e_{vu}, e_{uv}\} \in C'_i$ and 
    $\{f_{vw}, f_{wv}\} \in C'_j$. But by the definition of $G'$ and the fact that 
    $C'_i$ is a cycle, it must be the case that $C'_i$ covers $v_1 \in V'$
    or $v_2 \in V'$. Analogously, $C'_j$ must cover $v_1 \in V'$
    or $v_2 \in V'$ as well. In fact, by our preparatory observation, both cycles 
    cover both vertices $v_1$ and $v_2$ and hence they are not vertex-disjoint,
    a contradiction.
\end{proof}

Using Lemma~\ref{lemma:cycles_to_matchings}, we are now ready to provide the reductions \ECS\ $\leq_p$ \EWPM\ and \SOC\ $\leq_p$ \BCPM. Again, both reductions are very similar and we start with the former, repeating the construction that we already gave in Section~\ref{sec:related_cycle_problems}.

\begin{proof}[Proof of $\geq_p$ in Theorem~18]
We start by proving \ECS\ $\leq_p$ \EWPM\ and then move on to the proof of \SOC\ $\leq_p$ \BCPM. Again, the two reductions are very similar.

Assume that we are given an \ECS-instance $I$ consisting of $G = (V, E)$, $w : E \rightarrow \Z$, and $k \in \Z$.
We use the construction from Lemma~\ref{lemma:cycles_to_matchings} to obtain 
a graph $G'$ with edge-weights $w'$ and return the \EWPM-instance $I'$ consisting of $G' = (V', E')$, $w'$, and $k' \coloneqq k$.
Clearly, $G'$ and $w'$ can be found in polynomial time and it remains to 
prove correctness. 

Lemma~\ref{lemma:cycles_to_matchings} already does most of the work for 
us: If $I$ is a YES-instance, then there exist cycles $C_1, \dots, C_p$ 
of total weight exactly $k$ in $G$. By Lemma~\ref{lemma:cycles_to_matchings},
this implies the existence of $M$-alternating cycles $C'_1, \dots, C'_p$ (where 
$M$ is the perfect matching defined in Lemma~\ref{lemma:cycles_to_matchings}). 
The perfect matching $M^* \coloneqq M \Delta C'_1 \Delta \dots \Delta C'_p$ 
hence certifies that $I'$ is a YES-instance. 

Conversely, if $I'$ is a YES-instance, there must exist a perfect matching 
$M^*$ in $G'$ of weight exactly $k$. The symmetric difference $M \Delta M^*$ 
consists of vertex-disjoint $M$-alternating cycles. By Lemma~\ref{lemma:cycles_to_matchings},
this implies the existence of vertex-disjoint cycles $C_1, \dots, C_p$ in $G$. These cycles have the correct weight to certify that $I$ is a YES-instance.

This already concludes the proof of \ECS\ $\leq_p$ \EWPM. A small additional observation has to be made in order for 
the above construction to work for \SOC\ $\leq_p$ \BCPM: 
Without loss of generality, we can assume that $k$ is odd since we are looking for 
an odd cycle. As in \ECS\ $\leq_p$ \EWPM, 
the construction from Lemma~\ref{lemma:cycles_to_matchings} is now used to 
build the \BCPM-instance $I'$ consisting of $G'$, $w'$, and $k' = k$.

If $I$ is a YES-instance, there must exist a cycle $C$ in $G$ of 
odd weight $w(C) \leq k$. By Lemma~\ref{lemma:cycles_to_matchings}, this implies 
the existence of an $M$-alternating cycle $C'$ in $G'$ of the same weight,
i.e.\ $w'(C') = w(C) \leq k = k'$. The perfect matching $M^* \coloneqq M \Delta C'$ hence 
has weight exactly $w'(C')$ as $M$ has weight $0$. Since $w'(C')$ is odd and we have $w'(C') \leq k'$,
we conclude that $I'$ is a YES-instance.

Assume now for the other direction that $I'$ is a YES-instance. 
Then there exists a perfect matching $M^*$ in $G'$ of odd weight (since $k' = k$ is 
assumed to be odd) at most $k$. Looking at the symmetric difference $M \Delta M^*$,
we hence find an odd-weight $M$-alternating cycle $C'$ in $G'$. By Lemma~\ref{lemma:cycles_to_matchings},
this implies the existence of a cycle $C$ in $G$ of the same weight. Since 
we must have $0 \leq w'(C') \leq k' = k$, the same must be true for $w(C)$ and 
we conclude that $I$ is a YES-instance.
\end{proof}

\end{document}